\documentclass[letterpaper]{article} 
\usepackage{aaai22}  

\usepackage{times}  
\usepackage{helvet}  
\usepackage{courier}  
\usepackage[hyphens]{url}  
\usepackage{graphicx} 
\usepackage{natbib}  
\usepackage{caption} 
\usepackage{algorithm}
\usepackage{algorithmic}

\usepackage{newfloat}
\usepackage{listings}
\DeclareCaptionStyle{ruled}{labelfont=normalfont,labelsep=colon,strut=off} 
\floatstyle{ruled}
\newfloat{listing}{tb}{lst}{}
\floatname{listing}{Listing}
\usepackage{subcaption}
\usepackage{amsmath,amssymb,amsfonts, amsthm}
\usepackage{graphicx}
\usepackage{textcomp}
\usepackage{times}
\usepackage{color}
\usepackage{multirow}
\usepackage{url}
\usepackage[colorinlistoftodos,backgroundcolor=lightgray!50,linecolor=black,textsize=tiny]{todonotes}
\newtheorem{theorem}{Theorem}

\DeclareMathOperator*{\argmin}{arg\,min}
\DeclareMathOperator*{\argmax}{arg\,max}

\newcommand\myeq{\mkern1.2mu{=}\mkern1.2mu}
\newcommand\myplus{\mkern1.2mu{+}\mkern1.2mu}
\newcommand\myminus{\mkern1.2mu{-}\mkern1.2mu}

\newcommand\mynoteq{\mkern1.2mu{\neq}\mkern1.2mu}
\newcommand\myin{\mkern2mu{\in}\mkern 2mu}
\newcommand{\D}{\mathcal{D}}

\newcommand{\boldtheta}{\boldsymbol{\theta}}
\newcommand{\boldR}{\mathbf{r}}

\title{PATHFINDER: Designing Stimuli for Neuromodulation through data-driven inverse estimation of non-linear functions}

\author{
\thanks{This material is based upon work supported by the Naval Information Warfare Center (NIWC) Atlantic and the Defense Advanced Research Projects Agency (DARPA) under
Contract No. N65236-19-C-8017. Any opinions, findings and conclusions or
recommendations expressed in this material are those of the authors and do
not necessarily reflect the views of the NIWC Atlantic and DARPA.}\\
    Chaitanya Goswami$^{1}$, and
    Pulkit Grover$^{1}$
}
\affiliations{
    \textsuperscript{\rm 1} Electrical and Computer Engineering\\

    Carnegie Mellon University\\
    Pittsburgh, USA\\
}

\begin{document}

\maketitle
\begin{abstract}
There has been tremendous interest in designing stimuli (e.g. electrical currents) that produce desired neural responses, e.g., for inducing therapeutic effects for treatments. Traditionally, the design of such stimuli has been model-driven. Due to challenges inherent in modeling neural responses accurately, data-driven approaches offer an attractive alternative. The problem of data-driven stimulus design can be thought of as estimating an inverse of a non-linear ``forward" mapping, which takes in as inputs the stimulus parameters and outputs the corresponding neural responses. In most cases of interest, the forward mapping is many-to-one, and hence difficult to invert using traditional methods. Existing methods estimate the inverse by using conditional density estimation methods or numerically inverting an estimated forward mapping, but both approaches tend to perform poorly at small sample sizes. In this work, we develop a new optimization framework called PATHFINDER, which allows us to use regression methods for estimating an inverse mapping. We use toy examples to illustrate key aspects of PATHFINDER, and show, on computational models of biological neurons, that PATHFINDER can outperform existing methods at small sample sizes. The data-efficiency of PATHFINDER is especially valuable in stimulus design as collecting data is expensive in this domain.
\end{abstract}


%

\section{Introduction}\label{sec:introduction}
{N}{euromodulation} refers to altering neural activity\footnote{We use ``neural response'' and ``neural activity'' interchangeably.} through targeted delivery of a stimulus (e.g. electrical, chemical, ultrasound), and is one of the fastest growing areas of medicine, impacting millions of patients~\citep{krames2009neuromodulation}. {Many neurological disorders and diseases are a result of atypical neural activity in the brain and can be treated by providing appropriate stimuli which can ``correct'' this atypical neural activity.} In experiments, {controlling the neural activity through stimuli has shown promise} in treating Parkinsonian  symptoms~\citep{mastro2017cell,spix2021population}, stroke rehabilitation~\citep{cheng2014optogenetic},  regulating depression~\citep{chaudhury2013rapid}, etc. Indeed, the ability to systematically design stimuli that produce desired neural activity in the brain (which corrects the atypical activity) is key to treating several neurological disorders.

{Typically, a stimulus is characterized by a set of \textit{parameters}. For example, in~\citet{spix2021population}, the authors inject electrical currents (stimuli) in the brain to selectively stimulate a particular type of neuron (controlling neural activity) for treating Parkinsonian symptoms in mice. Their stimulus consisted of three parameters, namely, amplitude, frequency, and duration of the injected electrical currents. 
A common way to mathematically model the relation between the parameters of stimuli and the neural activity/response is through a ``forward mapping/function'' that takes as input the parameters and outputs the neural response~\cite{kandel2000principles}. Then, the problem of designing stimuli that produce desired neural responses reduces to \textbf{inverting the forward mapping}. Correspondingly, the set of stimulus parameters can be obtained by plugging in the desired neural response as an input to the inverse.}

{Broadly, there are three major challenges in estimating the inverse of the forward mapping. First, as the forward mapping depends upon the parameters being explored, for novel parameters, the forward mapping is generally unknown and needs to be estimated from the data~\cite{spix2021population}. Second, in most cases of interest, multiple parameter values lead to the same neural response. For example, in electrical stimulation of the brain, many stimuli produce the same neural firing rate and consequently, the same neural response~\cite{izhikevich2007dynamical}. This implies that the forward mapping in many cases is many-to-one, and hence, non-invertible. Therefore, instead of estimating an inverse, we need to estimate a pseudoinverse of the forward mapping (see Sec.~\ref{sec:prob_form} for more details). Third, the amount of data available in such healthcare settings is limited, and in general, the data-collection process is quite expensive (e.g. in~\citet{spix2021population}, authors could only collect dataset sizes of $\sim300$). So, it is desirable to estimate the pseudoinverse in a data-efficient manner.} 


{Recently, there has been significant interest in designing stimuli using data-driven methods~\cite{gonccalves2020training,bashivan2019neural,ponce2019evolving, walker2019inception, spix2021population} through some form of pseudoinverse estimation.} {Broadly, two methodologies have been explored in literature for estimating the pseudoinverse to design stimuli. One is to use conditional density estimation (CDE) methods to learn the conditional density of stimulus parameters conditioned on neural responses, and then use the conditional mode as the pseudoinverse~\cite{gonccalves2020training}. The other, what we call ``Naive Inverse'' (NI), is to estimate the forward mapping (e.g. using a neural network) and then numerically invert it~\cite{walker2019inception,ponce2019evolving}. We provide a detailed discussion regarding both approaches in Sec.~\ref{sec:background}. {Briefly, CDE-based methods are known to require \emph{more data than their regression counterparts}~\cite{holmes2012fast}, and NI approaches suffer because \emph{numerical inversion blows up even small errors in forward models (see Sec.~\ref{sec:discussion})}. These observations are reflected in our results in Sec.~\ref{sec:results}.}}

{In this work, we propose PATHFINDER: a novel pseudoinverse estimation framework that addresses the shortcomings of the two methodologies mentioned above. Specifically, PATHFINDER adapts regression techniques\footnote{Traditional regression techniques do not work for pseudoinverse estimation, see~\cite{chen2016nonparametric} for more details.} to \emph{directly} estimate a pseudoinverse, thereby circumventing the need of inverting an estimated forward model, while still requiring less data than CDE methods (since regression techniques typically require less data than their equivalent CDE counterpart). Sec.~\ref{sec:algo_des} provides a detailed description of PATHFINDER. The key insight utilized by PATHFINDER is that a non-invertible function can still be inverted over a restricted domain. If such a restricted domain were known \textit{a priori}, the inverse mapping can be estimated using traditional regression methods. PATHFINDER jointly learns a restricted domain and the inverse mapping over it. To do so, {PATHFINDER} uses a weighted $l_2$ loss, where the weights are also learned from data. On convergence, the weights approximate the indicator function over the restricted domain, effectively learning it (theoretically justified in Sec.~\ref{sec:results}). }


In Sec.~\ref{sec:results}, we compare the performance of PATHFINDER, with two CDE methods: Masked Autoregressive Flows (MAF)~\cite{papamakarios2017masked} and Mixture Density Networks (MDN)~\cite{bishop1994mixture}, as well as a naive inversion of a deep network (NI) in three toy examples, and a neuromodulation setup of electrically stimulating two neuron models. We quantify the performance of each method at different training dataset sizes and observe that PATHFINDER outperforms all the other methods for small dataset sizes, justifying our intuition discussed above. We discuss the results and limitations of our study in Sec.~\ref{sec:discussion}.

\section{Problem Statement and Notation}\label{sec:prob_form}


{We assume that each stimulus is characterized by $n$-different parameters, denoted as $\{\theta_i\}_{i=1}^n$, where $\theta_i\in\mathbb{R}$. Now, we define $\boldtheta=[\theta_1,\ldots,\theta_n]^T\in\mathbb{R}^n$ as the collection of all the $n$ parameters. $\Theta\subset\mathbb{R}^n$ denotes the collection of all allowed stimuli. Let the number of neural responses of interest be $m$ (e.g. in \citet{spix2021population}, authors aimed for selective stimulation between two neuron types, so $m=2$ for their case). Define, $\mathbf{r}=[r_{1},\ldots,r_{m}]^T \in\mathbb{R}^m$. Let $\mathcal{R}_{\Theta}$ be the collection of all distinct neural responses produced by all the stimuli $\theta\in\Theta$.}
For finding a {stimulus} that attains a desired {neural response}, we assume that we only have access to a dataset $\D\myeq\{\boldsymbol{\theta}_{i},\mathbf{r}_i\}_{i=1}^N$, formed by $N$ stimulus-response pairs.

    \textit{Problem Statement}: Given a dataset $\D = \{\boldsymbol{\theta}_{i},\mathbf{r}_i\}_{i=1}^N$, where $\{\boldsymbol{\theta}_{i}\}_{i=1}^N$ are independent and identically distributed (i.i.d.) samples from a distribution $p(\boldsymbol{\theta})$ on $\Theta$, and $\mathbf{r}_{i}$ is the neural response generated by the stimulus characterized by $\boldtheta_i$, the goal is to design/find a parameter $\boldtheta_{des}\in\Theta$ such that the stimulus corresponding to $\boldtheta_{des}$ can elicit (close to) a desired (user-specified) neural response $\mathbf{r}_{des}\in\mathcal{R}_{\Theta}$.

{Note that we only allow access to the (fixed) dataset $\D$. To restrict the scope of our work, we \emph{do not} allow the acquisition of more data, i.e., actively sampling based on inferences from existing data. All the methods discussed in this work can be extended to their active-sampling versions, and the performance of these methods in that setting will be explored in future work. We allow restricted access to neurons for fine-tuning some hyper-parameters of each method (see Sec.~\ref{sec:results}), but to compare fairly, this data is not used to train any of the methods. } 
{An important design parameter in our problem is choosing the set of allowed stimuli i.e. $\Theta$. This is decided \textit{a priori} by the user, typically based on domain knowledge, e.g.,~\cite{spix2021population} choose amplitude, frequency, and duration of the electrical waveforms as parameters of stimuli. For this work, we will assume that an appropriate $\Theta$ has already been chosen.}


We denote the forward mapping from the stimulus parameter space $\Theta$ to the neural response space $\mathcal{R}_{\Theta}$ as $g:\Theta\rightarrow\mathcal{R}_{\Theta}$. {For many-to-one functions such as $g$, a natural definition of a pseudoinverse is a mapping $g^{-1}:\mathcal{R}_{\Theta}\rightarrow\Theta_{inv}$, where $\Theta_{inv}\subseteq \Theta$ is a restricted domain such that $g(g^{-1}(\mathbf{r}))\myeq\mathbf{r}\ \forall\ \mathbf{r}\myin\mathcal{R}_{\Theta}$. E.g., $\cos:\mathbb{R}\rightarrow[-1,1]$ is many-to-one and not invertible, but restricting the domain of $\cos$ to $[0,\pi]$ helps define the familiar pseudoinverse:  $\cos^{-1}:[-1,1]\rightarrow[0,\pi]$. Note, for a pseudoinverse, equality in the other direction, i.e., $g^{-1}(g(\boldsymbol{\theta}))\myeq\boldsymbol{\theta}\ \forall\ \boldsymbol{\theta}\myin\Theta$ is not necessarily true, since the domain of $g^{-1}$ is $\Theta_{inv}$ (a subset of $\Theta$). A function can have multiple pseudoinverses (e.g., $\cos(\cdot)$, with $\Theta_{inv}\myeq[0\myplus 2n\pi,\pi\myplus2n\pi]$, has infinitely many pseudoinverses, one for each $n\myin\mathbb{N}$). For the goal of this work, estimating any one pseudoinverse suffices. We denote the estimated pseudoinverse of $g$ as 
$\widehat{g^{-1}}$. }

\section{Background and existing approaches}\label{sec:background}
{We now discuss natural approaches based on existing literature for this problem. }

{\textbf{Naive Inversion}: Conceptually, the simplest approach to estimating a pseudoinverse is what we call naive inversion (NI):  numerically invert an estimate $\widehat{g}(\cdot)$ of the forward mapping $g$. The general framework for NI is as follows:} 
\begin{align}
    \widehat{g} &= \argmin_{f\in\mathcal{F}_F} \sum_{(\boldsymbol{\theta}_i,\mathbf{r}_i)\in\mathcal{D}}l\left(f(\boldsymbol{\theta}_i),\mathbf{r}_i\right),\label{eq:naive_inversion:forward_loss}\\
    \widehat{g^{-1}}& = \argmin_{f\in\mathcal{F}_I} \sum_{(\boldsymbol{\theta}_i,\mathbf{r}_i)\in\mathcal{D}}l\left(\widehat{g}(f(\mathbf{r}_i)),\mathbf{r}_i\right),\label{eq:naive_inversion:inverse_loss}
\end{align} 
{where $\mathcal{F}_F$ and $\mathcal{F}_I$ denote the family of functions being considered for estimating the forward and inverse mapping (respectively), $\widehat{g^{-1}}$ is the estimate of pseudoinverse and $l(\cdot,\cdot)$ is an appropriate loss function (e.g. the $l_2$ loss)}.

{We now discuss a few recent works implementing the NI framework}. \citeauthor{ponce2019evolving} propose XDream, a genetic algorithm that fine-tunes the input visual stimuli to the generator of a pre-trained generative adversarial network~\cite{dosovitskiy2016generating} to maximize the response of neurons in the visual cortex.~\citeauthor{bashivan2019neural} train a convolutional neural network (CNN) using data collected from monkeys and use a softmax loss to find a visual stimulus that selectively activates a group of neurons in `V1' (the primary visual cortex).~\citeauthor{walker2019inception} develop a closed-loop experimental paradigm for optimizing visual stimulation in rats. They train a CNN over multiple data-collection sessions. After each session, they use the trained CNN to find the image that maximally excite the target neurons, through gradient ascent. Then, they retrain the CNN with these new images to find the images which maximally excite the target neurons. {A major drawback of the NI approaches is that~\eqref{eq:naive_inversion:inverse_loss} is often hard to optimize {(see Sec.~\ref{sec:discussion})}, as $\widehat{g}$ is typically a complicated function (e.g. deep neural network), which makes the overall loss prone to getting stuck in local minima. {This could explain why NI performs poorly in our numerical study in Sec.~\ref{sec:results}}. 

{\textbf{Conditional density estimation:} 
The modes of  the conditional distribution $p(\boldtheta|\boldR)$ are, by definition, the \textit{most likely} stimuli for producing the neural response $\boldR$. Thus, a natural candidate for the pseudoinverse of $g$ is:}  
\begin{align}
    \widehat{g^{-1}}(\boldR) := \argmax_{\boldsymbol{\theta}\in \Theta} {p}(\boldtheta|\boldR)\equiv \argmax_{\boldsymbol{\theta}\in \Theta} {p}(\boldtheta,\boldR), \label{eq:cde:inverse_mapping} 
\end{align}
where $p(\boldtheta,\boldR)$ denotes the joint density. {Two commonly-used CDE methods for estimating pseudoinverses are Masked Autoregressive Flow (MAF)~\cite{papamakarios2017masked} and Mixture Density Network (MDN)~\cite{bishop1994mixture}}}. We now discuss works that employ these approaches.~\citet{gonccalves2020training} uses Sequential Neural Posterior Estimation  (SNPE)~\cite{papamakarios2016fast} for estimating stimulus parameters in computational models of neurons from the visual cortex. SNPE uses MAF or MDN as a conditional density estimator. MAF is an instance of normalizing flows~\cite{papamakarios2019normalizing}. A particular drawback of normalizing flows (and hence MAF) is that we need to solve~\eqref{eq:cde:inverse_mapping} using an optimization technique (e.g. gradient descent~\cite{bishop2006pattern}), which may get stuck in a local minima. 
\citet{unni2020deep}, and \citet{zen2014deep}, though not focused on neural stimulation, also use MDNs to estimate pseudoinverses. MDNs with Gaussian component distributions are particularly attractive for learning pseudoinverses as the means predicted by the MDN provide an approximate estimate of the modes, thus not requiring us to solve~\eqref{eq:cde:inverse_mapping}. 

Broadly, CDE-based approaches require a large amount of data and are hard to implement in high dimensions~\cite{papamakarios2019sequential}. Regression methods, although still suffering from the curse of dimensionality, can perform better than their CDE counterparts at a lower number of data samples~\cite{holmes2012fast}, which we also observed in our simulation results (see Sec.~\ref{sec:results}). One regression-based approach is PATHFINDER, proposed here and discussed next. In Sec.~\ref{sec:results}, we compare MDN, MAF, NI, and PATHFINDER, in toy examples and a neuromodulation context.  
\section{PATHFINDER}\label{sec:algo_des}
{PATHFINDER estimates a pseudoinverse} by exploiting the insight, as discussed in Sec.~\ref{sec:prob_form}, that many-to-one functions can be inverted over appropriately restricted domains. If such a restricted domain $\Theta_{inv}$ were known \textit{a priori}, then a restricted dataset could be created by excluding all the data points $(\boldtheta_i,\boldR_i)$ where {$\boldtheta_i\not\in\Theta_{inv}$ from} the dataset $\D$. As $g$ is invertible over $\Theta_{inv}$, any traditional regression technique applied to this restricted dataset would yield a pseudoinverse corresponding to $\Theta_{inv}$. Formally, a pseudoinverse on $\Theta_{inv}$ could be estimated as: 
 \begin{equation}\label{eq:pathfinder:revised_loss}
     \widehat{g^{-1}}=\argmin_{f\in\mathcal{F}}\frac{1}{N}\sum_{i=1}^N\mathbb{I}[\boldsymbol{\theta}_i\in\Theta_{inv}]\|f(\mathbf{r}_{i})-\boldtheta_i\|_2^2,
\end{equation}
where $\mathbb{I}(\cdot)$ is the indicator function and $\mathcal{F}$ is the family of functions used for regression. The challenge is that we only have access to a dataset (and not the forward mapping), so $\Theta_{inv}$ is not known \textit{a priori}. To address this, PATHFINDER \textit{jointly} estimates \textit{both} a restricted domain and the corresponding pseudoinverse as follows:   
  \begin{align}
\widehat{g^{-1}},\{\widehat{w}(\boldsymbol{\theta}_i)\}_{i=1}^N&\myeq\argmin_{\substack{f\in\mathcal{F},\\w(\boldsymbol{\theta}_i)\geq 0}}\frac{1}{N}\sum_{i=1}^N\underbrace{w(\boldsymbol{\theta}_i)\|f(\mathbf{r}_{i})-\boldsymbol{\theta}_i\|_2^2}_{loss}\myplus\nonumber\\ &\underbrace{\beta \sum_{i=1}^Nw^2(\boldsymbol{\theta}_i)}_{regularizer}, s.t.\  \frac{1}{N}\sum_{i=1}^Nw(\boldsymbol{\theta}_i)\myeq1,\label{eq:pathfinder:main_loss}
\end{align}
where $\beta\myin\mathbb{R}^{+}$ is a hyper-parameter. Equation~\eqref{eq:pathfinder:main_loss} follows the philosophy of \eqref{eq:pathfinder:revised_loss},  approximating $\mathbb{I}[\boldsymbol{\theta}_i\myin \Theta_{inv}]$ in \eqref{eq:pathfinder:revised_loss} by $\widehat{w}(\boldsymbol{\theta}_i)$, which are learned jointly with $\widehat{g^{-1}}$. 

How does the PATHFINDER optimization~\eqref{eq:pathfinder:main_loss}  incentivize learning of a restricted domain? If only  parameters belonging to a restricted domain have non-zero weights $(\widehat{w}(\boldsymbol{\theta}_i))$, the loss term would be low because the corresponding inverse mapping can be estimated accurately. Hence, the loss term encourages PATHFINDER to learn weights that are non-zero only for  some restricted domain over which $g$ is invertible.
 It is desirable for {an estimate of the} pseudoinverse $\widehat{g^{-1}}$ (as discussed in Sec.~\ref{sec:prob_form}) to have as its domain the entire $\mathcal{R}_{\Theta}$, or at least as large a subset as possible, so it can provide a stimulus for as many neural responses as possible. This implies that for PATHFINDER to estimate a pseudoinverse, the image of the restricted domain learned by it should be as large as possible. This condition is not ensured by the loss term in \eqref{eq:pathfinder:main_loss}, since it is small for any restricted domain (e.g., consider two restricted domains $[0,1]$ and $[0,\pi]$ of $\cos(\cdot)$.  $[0,\pi]$ is more desirable here. The loss term in \eqref{eq:pathfinder:main_loss} is small for both domains, and is unable to discriminate between them). 

To encourage PATHFINDER to learn a restricted domain with a large image, we use the following observation: The largest restricted domain $\Theta_{max}$ (measured by its total probability under $p(\boldtheta)$) over which $g$ can be inverted, also has the largest image $g(\Theta)=\mathcal{R}_{max}$ (measured under $p(\boldtheta)$) among all such invertible restricted domains. I.e., the image of the largest restricted domain is equal to $\mathcal{R}_{\Theta}$, up to a set of zero probability (proof in Supp. Sec. 2). 
This observation is incorporated in the regularizer and the constraint in \eqref{eq:pathfinder:main_loss} to encourage PATHFINDER to learn as large a restricted domain as possible. To see this, let us analyze the following optimization that distills the effect of the regularizer for distinguishing among restricted domains over which $g$ is invertible 
(the first term in \eqref{eq:pathfinder:main_loss} is  low for such 
domains): 
$\{w^*_i\}_{i=1}^N=\argmin_{\{w_i\}_{i=1}^N}\sum_{i=1}^Nw_i^2,\text{s.t.}\ \frac{1}{N}\sum_{i=1}^Nw_i=1,\nonumber  \sum_{i=1}^N\mathbb{I}[w_i\mynoteq 0]=K,\  w_i\geq0\ \forall\ i.$
This optimization explores the behaviour of the regularizer when only $K$  out of $N$ total weights are non-zero, and has the solution: $w^*_i\myeq N/K$ for any $K$ out of $N$ weights and the rest are 0. Hence, the regularizer term in \eqref{eq:pathfinder:main_loss} scales approximately as $\sim 1/K^2$ for $K$ non-zero weights, which incentivizes making a larger number of $\widehat{w}(\theta_i)$ to be non-zero, encouraging PATHFINDER to consider as large a restricted domain as possible. 

Thus, there is a careful interplay between the loss, the regularizer, and the constraints in \eqref{eq:pathfinder:main_loss}. The loss encourages learning non-zero weights ($\widehat{w}(\boldsymbol{\theta}_i)$) only over a restricted domain; the regularizer and the constraints try to make the restricted domain as large as possible. A desirable pseudoinverse can be learned by carefully choosing the value of $\beta$.


\section{Results}\label{sec:results}

\subsection{Theoretical Result}\label{sec:theory}
We provide a formal justification for the intuition behind PATHFINDER (discussed in Sec.~\ref{sec:algo_des}), albeit under idealized assumptions of (sufficiently rich $\mathcal{F}$, noiseless data). {The optimization problem of PATHFINDER defined in \eqref{eq:pathfinder:main_loss} can be viewed as approximating the following problem:}
\begin{align}
g^{-1*},w^*\myeq \argmin_{f\in\mathcal{F},w\in\mathbb{W}}&\mathbb{E}_{p(\boldtheta)}\left[w(\boldsymbol{\theta})\|f(\mathbf{r})-\boldsymbol{\theta}\|_2^2\myplus\beta w^2(\boldsymbol{\theta})\right],\nonumber\\&\ s.t.\  \mathbb{E}_{p(\theta)}\left[w(\boldsymbol{\theta})\right]\myeq1,\label{eq:algorithmDes:exactlossform}
\end{align}
{where $\mathbb{W}=\{w:\Theta\rightarrow\mathbb{R}^{+}|\ w$ is a measurable function\}, $\mathcal{F}$ is the family of the functions being considered for regression and $\mathbb{E}_{p(\boldsymbol\theta)}[\cdot]$ is the expectation with respect to $p(\boldsymbol{\theta})$ (defined in Sec.~\ref{sec:prob_form})}. In the following theorem, $g^{-1}_{max}$ is the pseudoinverse corresponding to $\Theta_{max}$, the largest restricted domain over which $g$ is invertible (see Sec.~\ref{sec:algo_des}).
 
\begin{theorem}\label{onlyTheorem} Assume that  $g:\Theta\rightarrow\mathcal{R}_{\Theta}$ is a Lipschitz $l_2$-integrable function such that $g^{-1}_{max}$ as defined above exists, $g^{-1}_{max}\in\mathcal{F}$ where $\mathcal{F}$ is the family of functions being considered for estimation, and the dataset is noiseless. Then, for any $\epsilon>0,
\ \exists$ an $A_{\epsilon}>0$ such that for $0<\beta\leq A_{\epsilon}$:  
\begin{align}
&\mathbb{E}_{p(\boldsymbol\theta)}\left[\|g^{-1*}(\mathbf{r})-g^{-1}_{max}(\mathbf{r})\|_2^2\right] \leq  c_1\epsilon+c_2\sqrt{\epsilon},\label{eq:theory:result}
\end{align}
where $g^{-1*}$ is the solution of \eqref{eq:algorithmDes:exactlossform}, {and $c_1,c_2\in\mathbb{R}^{+}$.}
\end{theorem}
\begin{proof}
The proof is provided in Supp. Sec.~3.
\end{proof}
 Theorem \ref{onlyTheorem} implies that if the global optimum of \eqref{eq:algorithmDes:exactlossform} is attainable, then given enough data, PATHFINDER can estimate the pseudoinverse $g^{-1}_{max}$ with arbitrary precision by tuning $\beta$. Note, as the optimization problem defined in \eqref{eq:algorithmDes:exactlossform} is non-convex, guaranteeing convergence to this optimum is non-trivial, but in practice, stochastic gradient descent methods~\cite{bishop2006pattern} performed reasonably well in solving the PATHFINDER loss.

\subsection{Simulation Results: Implementation Details}
We compared the performance of PATHFINDER with 3 competing techniques: MAF, MDN, and NI (described in Sec.~\ref{sec:background}) in 3 toy examples and a neuromodulation setup of electrically stimulating neuron models. A brief description of implementation details is given below with more details provided in Supp. Sec. 5.

\textbf{Splitting the dataset}:
Since the input to our model is the neural response $\mathbf{r}_i$ and not the stimulus parameter $\boldtheta_i$, we split our data into training, test, and validation sets in the following manner:
Split all possible neural responses present in $\D$ into training, validation, and test neural responses. Let $\mathcal{R}_V$, $\mathcal{R}_{Tr}$, and $\mathcal{R}_{Te}$ be the sets containing the validation, training, and test neural responses, respectively. Remove all the $(\boldtheta_i,\boldR_i)$ from the original dataset $\D$ where $\boldR_i\in \mathcal{R}_{Te}\cup\mathcal{R}_V$, i.e., $\boldR_i$ is present in the test and validation set, to construct the training dataset $\D_{Tr}$. Note that for any $\boldR\in \mathcal{R}_{Te}\cup\mathcal{R}_V$, there may be multiple stimulus parameters $\boldtheta$ in the original dataset which produce $\boldR$, and we remove all of them. The validation dataset $\D_{V}$ can be constructed similarly by removing $(\boldtheta_i,\boldR_i)$ from the original dataset $\D$ where $\boldR_i\in \mathcal{R}_{Tr}\cup\mathcal{R}_{Te}$. For the test set, we only store the neural responses, i.e $\mathcal{R}_{Te}$, as we want to generate stimuli that produce those neural responses.

\begin{table*}[t]
    \centering
    \begin{tabular}{||c|c|c|c||}
    \cline{2-4}
    \multicolumn{1}{c|}{}&\multicolumn{3}{|c|}{\textbf{MAE for toy forward mappings}}\\
    \hline
    \textbf{Techniques}& $r=e^{-\theta^2/2}; N=10$&$r=\cos(\theta); N=20$ &$r=(\theta^2-4)^2; N=50$\\  
    \hline
    \hline
     {PATHFINDER}    & $\mathbf{7.2}\%\pm\mathbf{2.1}\%$&$\mathbf{14.3}\%\pm\mathbf{3.1}\%$&$\mathbf{10.5}\%\pm\mathbf{1.8}\%$\\
         \hline
     {MDN}    &$10.1\%\pm2.7\%$& $20.8\%\pm 4.2\%$&$24.0\%\pm10.5\%$\\
         \hline
     {MAF}    & $29.2\%\pm3.5\%$&$31.8\%\pm3.9\%$&$22.0\%\pm2.0\%$\\
         \hline
     {NI}  &  $20.9\%\pm3.1\%$& $18.8\%\pm3.6\%$&$21.1\%\pm25.2\%$\\
    \hline
    \end{tabular}
    \caption{The table contains the NMAE (see Sec.~\ref{sec:results}) values for PATHFINDER, MAF, MDN, and NI for the three different toy examples corresponding to the simulation study discussed in Sec.~\ref{sec:sim_result:toy_example}. $N$ denotes the dataset size at which the NMAE values were calculated. The NMAE values are averaged across $50$ independent trials, and are presented with their $99\%$ confidence intervals.}
    \label{tab:1:toy_example}
\end{table*}
\textbf{Evaluating the Validation/Test Loss and figure of merit}: We will explain our figure of merit by taking an example of calculating it over the validation set. For every neural response $\boldR\in\mathcal{R}_V$, we obtain the corresponding $\widehat{\boldsymbol{\boldtheta}}$. We feed the $\widehat{\boldsymbol{\boldtheta}}$ to the neuron/model to obtain its actual neural response $\mathbf{r}^{act}$. Since we have $m$ different neural responses, we calculate normalized mean absolute error (NMAE) for each neural response. Let the maximum and minimum values that can be achieved for the $i$-th neural response be $r_{max}^{i}$ and $r_{min}^{i}$, i.e.
$r_{max}^{i} = \max_{\boldR\in\mathcal{R}_{\Theta}} e_i^{T}\boldR;\ \ \      r_{min}^{i} = \min_{\boldR\in\mathcal{R}_{\Theta}} e_i^{T} \boldR,$
where $e_i\in\mathbb{R}^m$ with $1$ at the $i$-th dimension and $0$ everywhere else. Then, we define NMAE for the $i$-th neural response as:  
 \begin{align}
     \text{NMAE} &= \frac{100}{|\mathcal{R}_V|(r_{max}^{i}\myminus r_{min}^{i})}\sum_{\boldR\in\mathcal{R}_V}\left|e_i^T\boldR^{act}\myminus e_i^T\boldR \right|,\label{eq:results:mae}
 \end{align}
 where $|(\cdot)|$ is the absolute function and $g\left(\widehat{g^{-1}}(\boldR)\right)=\boldR^{act}$. The NMAE quantifies how close the neural response produced by the predicted stimulus is to the desired neural response on a scale of $0$ to $100$. Similarly, the test NMAE can be calculated by replacing the validation set with the test set. Notice that, for calculating the NMAE, we require access to the neuron (hence the need for accessing the neuron during hyper-parameter tuning, see Sec.~\ref{sec:prob_form}).

\textbf{Implementation details of data-driven approaches}:

\textbf{MAF}: We used the implementation suggested in the original work~\cite{papamakarios2017masked}. MAF was trained to learn the joint density of data $p(\boldtheta,\boldR)$ by minimizing the negative log-likelihood (both joint and conditional densities are equivalent for estimating the pseudoinverse (see~\eqref{eq:cde:inverse_mapping}), and joint density contains more information). For toy examples, we used MAF with 3 flows, and for neuron models, we varied the number of flows from 2 to 8.
For all simulations, the corresponding Masked Autoencoder for density estimation~\cite{germain2015made} had 2 hidden layers. After each flow, a batch normalization flow (see~\citet{papamakarios2017masked}) was used. The order of the inputs was reversed after each flow. The initial order of the inputs was assigned randomly. The base density used was a standard Gaussian. The conditional mode was calculated by solving~\eqref{eq:cde:inverse_mapping} with the Adam optimizer~\cite{kingma2014adam}, with a learning rate of $10^{-3}$. 

\textbf{MDN}: We used the formulation discussed in the original work~\cite{bishop1994mixture} with minor adaptations. The component distribution was Gaussian with diagonal covariance. For all toy examples, 5 hidden layers were used. The number of mixtures used was $15,2$, and $10$ for $\cos(2\pi\theta)$, $e^{-\frac{\theta^2}{2}}$, and $(\theta^2-4)^2$, respectively. For the neuron models, the number of mixtures was varied from 10 to 100 in intervals of 10, and the number of hidden layers was varied from 5-8. The weight and mean output layer were followed by softmax and linear activation, respectively. For the variance output layer, we used the activation: $ELU(X)+1$, where $ELU(\cdot)$ is the exponential leaky unit~\cite{clevert2015fast}. The variance of each mixture was clipped at $10^{-4}$ for numerical stability. The conditional mode was calculated as the conditional mean of the most likely component of the mixture. 

\textbf{Naive Inversion:} The architecture can be visualized as an autoencoder, where the decoder serves as the forward model and the encoder as the inverse. For the toy examples, both the forward and inverse models had 5 hidden layers. For the neuron model case, the number of hidden layers was varied from 4-8, for both the forward and inverse mappings. 

\textbf{PATHFINDER:} We used two fully connected networks, one for estimating the inverse mapping ($\widehat{g^{-1}}$) (regressor network) and the other for estimating the weight mapping ($\widehat{w}$). In toy examples, both networks consist of 5 hidden layers. For the neuron models, the number of layers was varied from 5-8 for both the regressor and the weight network. The output layer of the regressor and the weight network has linear activation. For the toy examples, $\beta$ was chosen to be $10^{-4}, 0.01$, and $0.00005$ for the $\cos(2\pi\theta)$, $e^{-\frac{\theta^2}{2}}$, and $(\theta^2-4)^2$ mapping. For the neuron models, we searched across the following values of $\beta$: $0.1$, $0.01$, and $0.001$. 


$\mathrm{ReLU}$ activation was used for the hidden layers of all approaches. The number of hidden units for each hidden layer in toy examples and neuron models is $10$ and $100$, respectively. While training PATHFINDER, a simplex projection \cite{chen2011projection} was performed on the batch output of the weight network to ensure the constraints specified in \eqref{eq:pathfinder:main_loss} were met. Adam optimizer with a learning rate of $0.001$ is used for minimizing the objectives for all techniques. Batch normalization \cite{ioffe2015batch} was applied wherever applicable. Different batch sizes were used for different sizes of the training datasets, and are provided in Supp. Sec 5. We trained each model until the convergence of validation loss (note \emph{loss} instead of MAE), defined as the relative change in the validation loss between successive steps being less than 0.1\%. 
For selecting the hyper-parameters of each model ($\beta$ in PATHFINDER, etc.), we calculated the validation NMAE (which we found to be better than the validation loss). A detailed description of all hyper-parameters is provided in Supp. Section 5. The implementation of all the techniques, i.e. PATHFINDER, MDN, MAF, and NI, was performed in \texttt{TensorFlow}~\cite{tensorflow2015-whitepaper} and \texttt{TensorFlow probability}~\cite{dillon2017tensorflow}.

\subsection{Simulation Results: Toy Examples}\label{sec:sim_result:toy_example}

\textbf{Setup}:
We considered three different toy forward mappings for estimating pseudoinverses: $ r\myeq\cos(2\pi\theta)\myplus\epsilon,\ \theta\in[0,3];\  r\myeq(\theta^2-4)^2\myplus\epsilon,\ \theta\in[-3,3];\ r\myeq e^{-\frac{\theta^2}{2}}\myplus\epsilon,\ \theta\in[-3,3],$  where $\epsilon\sim\mathcal{N}(0,0.01)$ (Gaussian distribution with mean $0$ and variance $0.01$). To create a dataset of size $N$ for a toy mapping, i.e. $\D=\{\theta_i,r_i\}_{i=1}^N$, $N$ samples of $\theta$ were uniformly randomly sampled from that toy forward mapping's respective domain and substituted into each toy model to obtain the  corresponding $\{r_i\}_{i=1}^N$. 

To characterize the data-efficiency of each technique, i.e. MDN, MAF, NI, and, PATHFINDER, in estimating the pseudoinverses of these toy mappings, we performed the following study: For each toy forward mapping, we start with a training dataset having $10$ datapoints and keep increasing the dataset in increments of $10$ datapoints until one of the techniques produces NMAE less than $15\%$ (the threshold of 15\% was decided arbitrarily). We record the NMAE for all four techniques at this particular training dataset size. The corresponding dataset sizes at which the NMAE were calculated for each toy mapping are $20$ for $\cos(\theta)$, $10$ for $e^{-\frac{\theta^2}{2}}$, and $50$ for $(\theta^2-4)^2$. 

\textbf{Results}: The NMAE for all the four techniques and the three different toy mappings are listed in Table~\ref{tab:1:toy_example}.  We observe that, PATHFINDER has the smallest NMAE among all the competing methods across all the toy forward mappings. This result, demonstrates that PATHFINDER requires the least amount of data to estimate pseudoinverse with reasonable accuracy. {Results remain qualitatively the same when the noise variance is increased to $10\%$ of the range of each mapping (see Supp. Sec. 7)}

\subsection{Simulation Results: Electrical Stimulation in Neuron Models}\label{sec:sim_result:neuron}
\textbf{Setup}: The setup of the simulation is detailed in Supp. Sec. 4. Briefly, we explore the use of electrical waveforms to stimulate an excitatory pyramidal neuron (Pyr neurons) and an inhibitory parvalbumin-expressing neurons (PV neurons). Modulating the firing rate (neural response) of excitatory and inhibitory neurons using electrical currents has clinical relevance, e.g., in stopping seizures~\cite{avoli2016specific, mahmud2016differential}. We adapt the simulation setup of~\citet{mahmud2016differential}. We replace the single-compartment neuron models used in~\citet{mahmud2016differential} with more realistic
multi-compartment neuron models taken from the Allen Cell Type Database~\cite{de2020large}. Both neurons were simulated in the NEURON software~\cite{carnevale2006neuron}, using the \texttt{Allen SDK} package~\cite{de2020large} in \texttt{python}~\cite{10.5555/1593511}. The stimulus parameters that describe the waveforms were chosen to be the coefficients of $50$ sinusoids, namely, $u_{\boldtheta}(t) = \sum_{i=1}^{50}\theta_i\sin(2\pi (i-1) t)$, where $\boldtheta = \begin{bmatrix}\theta_1&\hdots&\theta_{50}\end{bmatrix}.$ The duration of the waveform was fixed at 200 ms. Parameters ($\boldtheta_i$'s) were randomly sampled from a uniform hyper-sphere of radius 2 nA to encourage a more uniform distribution of power across waveforms. The corresponding firing rate, defined as $\frac{\text{Total Number of Spikes}}{200 \text{ ms}}$, of each neuron model was recorded. We define our neural response as $\mathbf{r}\myeq[r_{1},r_{2}]^T$, where $r_{1}$ and $r_{2}$ are the firing rates for the Pyr and the PV neuron, respectively. 


We calculated the average NMAE for PATHFINDER, MDN, MAF, and NI across $50$ independent trials for training dataset sizes: $50, 100, 250,500,$ and $1000$. The training dataset sizes were chosen to reflect the actual dataset sizes reported in the literature. Typically, such datasets are collected using in-vitro patch-clamp electrophysiology~\cite{perkins2006cell} which is extremely expensive. Publications employing these techniques usually report dataset sizes of around 200-300 samples (waveforms)~\cite{spix2021population, mastro2017cell}.

\textbf{Results}: Fig.~\ref{fig:4:neuron}\textbf{c} and Fig.~\ref{fig:4:neuron}\textbf{d} show the individual NMAE for the PV and Pyr models, respectively. Fig.~\ref{fig:4:neuron}\textbf{a} and Fig.~\ref{fig:4:neuron}\textbf{b} show the average and the maximum of the NMAE across the two neuron types. As with the toy examples, we observe that PATHFINDER significantly outperforms the other techniques for low sample size (Fig~\ref{fig:4:neuron}\textbf{a} and Fig~\ref{fig:4:neuron}\textbf{b}). The gains are even better than the toy example, as even at $1000$ training datapoints (as compared to the $50$ datapoints in the toy example), PATHFINDER continues to outperform the other techniques. MAF performs better than MDN for the neuron model case. This result is aligned with the result from the original works for MAF~\cite{papamakarios2019sequential}. This aspect is different for our toy examples (Table~\ref{tab:1:toy_example}), where MDN performs better than MAF. NI seems to perform the worst out of all the methods tested for the neuron model case. This is in contrast to its good performance in the toy example case. We discuss these observations in more detail in Sec.~\ref{sec:discussion}. 

\begin{figure}[t]
\centering
\includegraphics[scale=0.39]{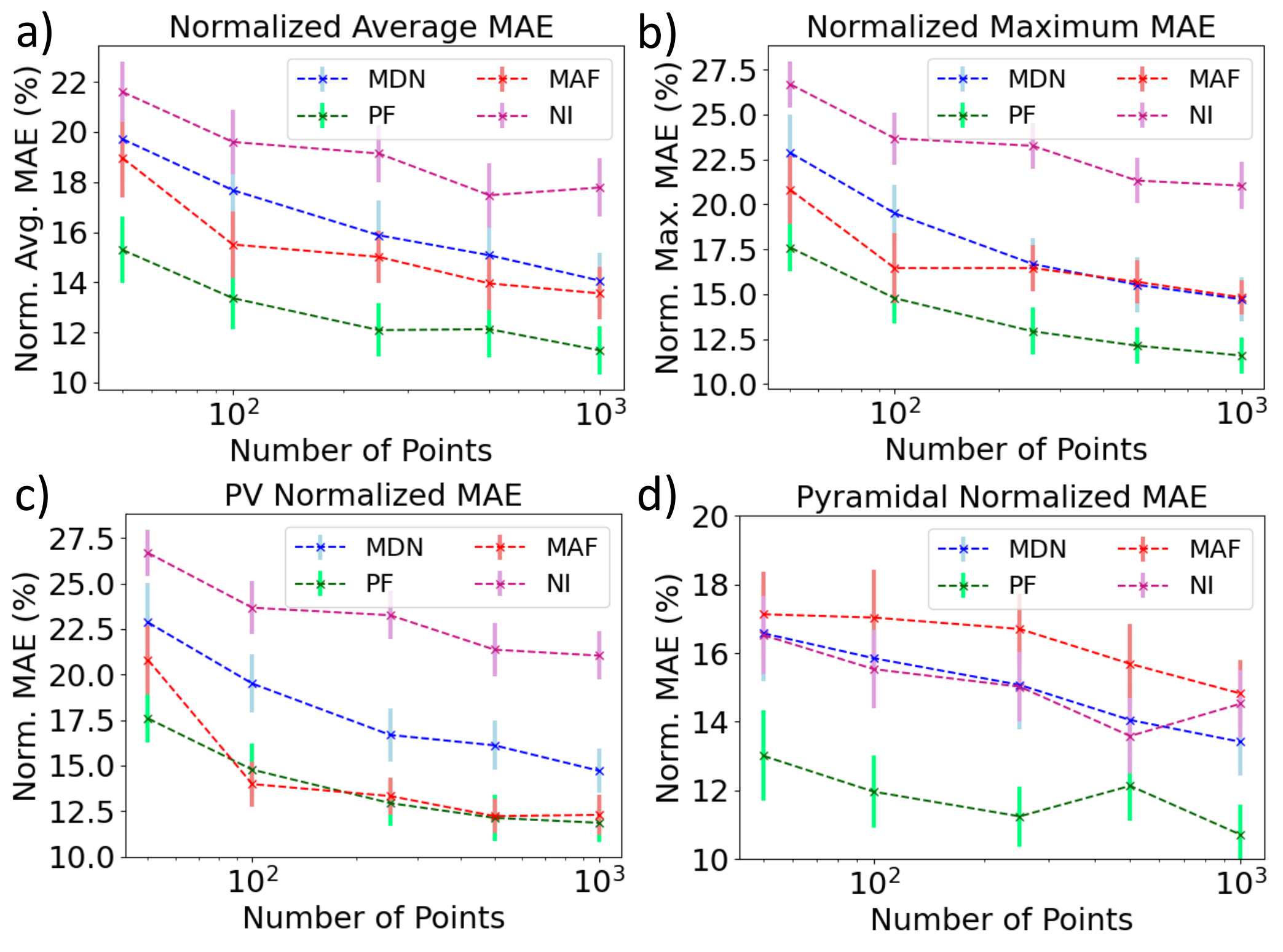}

\caption{Fig.~\ref{fig:4:neuron}\textbf{a} and~\ref{fig:4:neuron}\textbf{b} show the average and maximum of the NMAE across both neuron types obtained by PATHFINDER (PF), MDN, MAF, and NI at different training dataset sizes for the neuron model example (see Sec.~\ref{sec:sim_result:neuron}). Fig.~\ref{fig:4:neuron}\textbf{c} and Fig.~\ref{fig:4:neuron}\textbf{d} show the individual NMAE for PV and Pyr neurons, respectively, for all the 4 techniques at different training dataset sizes. The values of NMAE at each training dataset size were averaged across $50$ independent trials and are shown with a $99\%$ confidence interval.}
\label{fig:4:neuron}
\end{figure}

\section{Discussion and Limitations}\label{sec:discussion}
\begin{figure*}[t]
\centering
 \includegraphics[scale=0.43]{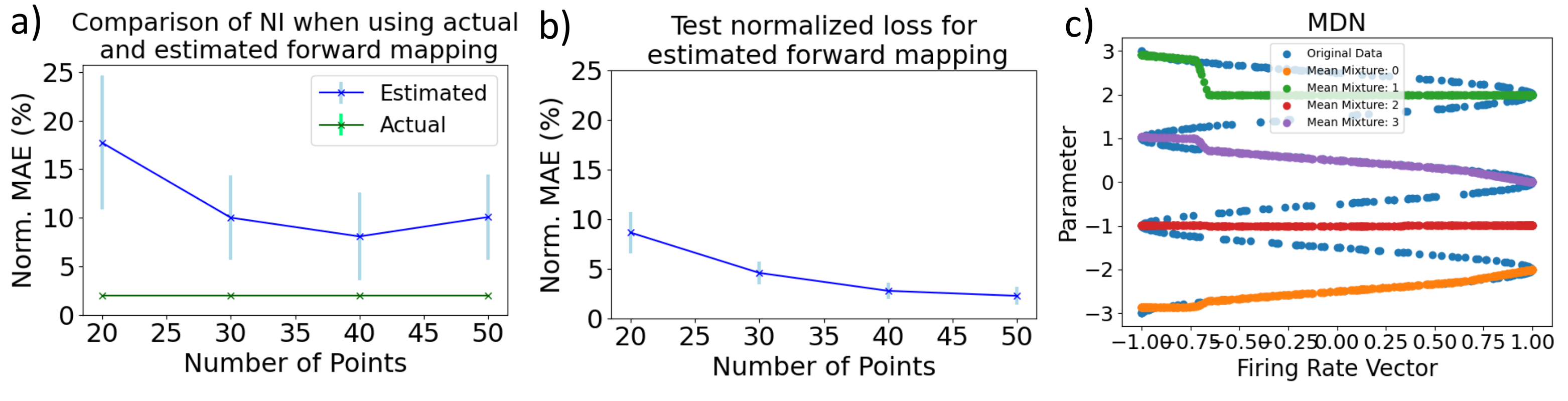}

\caption{Fig.~\ref{fig:7:misc}\textbf{a} represents the average test NMAE for the NI framework with estimated and actual forward mapping across 50 different trials (see Sec.~\ref{sec:discussion}). Fig.~\ref{fig:7:misc}\textbf{b} presents the test loss for the estimated forward mapping of the $\cos(\cdot)$ mapping in NI. The color bars are the $99\%$ confidence interval. Fig~\ref{fig:7:misc}\textbf{c} shows the conditional means estimated by the MDN for the $\cos(\theta)$ mapping.}
\label{fig:7:misc}
\end{figure*}
\textbf{Data-Efficiency of PATHFINDER}: A potential reason for the data-efficiency of PATHFINDER (seen in Sec.~\ref{sec:results}) is what we call the \emph{maximization bias}. Maximization bias is the phenomenon that PATHFINDER tends to estimate the pseudoinverse with the largest number of datapoints. E.g., in $\cos(\cdot)$ mapping, assume that PATHFINDER only estimates 1 of the 6 pseudoinverses corresponding to the restricted domains $[0,0.5]$,  $[0.5,1]$, $[1,1.5]$, $[1.5,2]$, $[2,2.5]$ and $[2.5,3]$. Let $n_i$ be the number of datapoints that lie in the restricted domain of the $i$-th pseudoinverse. Recall that the regularizer in the PATHFINDER loss tries to give non-zero weights to as many datapoints as possible (Sec.~\ref{sec:algo_des}). Consequently, the regularizer encourages PATHFINDER to choose the restricted domain having the largest number of datapoints, i.e., $\max_{i} n_i$. 

For a size-$n$ dataset sampled using a uniform distribution, we show (in Supp. Sec.~6) that $\mathbb{E}\left[\max_{i} n_i\right] = \frac{n}{6}+c\sqrt{n}$ ($\mathbb{E}[\cdot]$ is the expectation, and $c$ is a constant). Note that, while the \textit{expected} number of datapoints in any one restricted domain is $n/6$, {the largest restricted domain has extra $c\sqrt{n}$ datapoints. PATHFINDER loss  encourages inversion in just this restricted domain, and thus, is able to harness these extra datapoints thereby lowering the required overall sample size}. On the other hand, MDN, MAF, and NI, in one way or another, try to estimate the whole forward mapping, and are not designed to harness maximization bias. A more detailed explanation is provided in Supp. Sec. 6. We aim to precisely characterize the maximization bias in future works.



\textbf{MDN v/s MAF}:  The number of modes of $p(\boldtheta|\boldR)$, {in most cases of interest}, is determined by the number of parameters $\boldtheta$ that produce the response $\boldR$, e.g. $p(\boldtheta|\boldR)$ for $\cos(\cdot)$ forward mapping has $6$ modes. For all our toy examples, $p(\boldtheta|\boldR)$ has less than $6$ modes but in neuron models, $p(\boldtheta|\boldR)$ has an extremely large number of modes (technically, infinite), due to the dimension of $\boldtheta$ being higher than $\boldR$. At small sample sizes, {MDNs, and in general mixture models, are known to overfit  $p(\boldtheta|\boldR)$ when the number of modes is large~\cite{davis2020use}. This explains why MDNs work well in our toy examples, but not in our neuron models. On the other hand, MAF, being a normalizing flow estimator, avoids overfitting due to a large number of modes and performs better than MDN in neuron models.}



For a $p(\boldtheta|\boldR)$ with large number of modes, MDN can still estimate pseudoinverses reasonably well, as evidenced from our neuron model study, possibly due to \textit{mode collapse}~\cite{theis2015note}. For estimating pseudoinverses, we do not need MDNs to model every mode, but rather just one. {E.g., Fig.~\ref{fig:7:misc}\textbf{c} illustrates that MDNs having fewer mixtures than the total number of modes of $p(\boldtheta|\boldR)$, can still estimate $g$'s pseudoinverse.} Therefore, mode collapse actually helps, not hinders, in estimating a pseudoinverse. However, in MDNs, mode collapse is not ``controlled'' explicitly, and hence on average, performance of MDN suffers. PATHFINDER can be viewed as ``controlled'' mode collapse to learn the pseudoinverse corresponding to the largest restricted domain.

\textbf{Failure of NI in neuron models}: 
To understand the degradation in the performance of NI from our toy examples to the neuron-model examples, we performed the following experiment: {For the $\cos(\cdot)$ mapping, we compared the performance in estimation of $\widehat{g^{-1}}$ in i) traditional NI (where we use the neural network estimate $\widehat{g}$ of $g$ in~\eqref{eq:naive_inversion:inverse_loss}), with ii) using the \textit{actual} mapping $g(\cdot)$ instead of $\widehat{g}$ in~\eqref{eq:naive_inversion:inverse_loss} ($\cos(\cdot)$ in this case)}. The plot of NMAE vs training dataset sizes is shown in Fig.~\ref{fig:7:misc}\textbf{a}. NI using the actual mapping $g(\cdot)$ achieves a small error for even a small number of data points, whereas NI with estimated forward mapping $\widehat{g}$ does not. {The latter observation is surprising because}, from the test loss for the estimated forward mapping (shown in Fig.~\ref{fig:7:misc}\textbf{b}), it might seem that the estimated forward mapping is a good fit (especially at $40$ and $50$ samples). Together, Fig.~\ref{fig:7:misc}\textbf{a} and~\ref{fig:7:misc}\textbf{b} illustrate the challenge of numerically inverting a deep network: {even small inaccuracies in the forward model lead to substantial error upon inversion}. {Worse}, in our neuron models example, the estimated forward mapping has a larger error than toy examples {($10.5$\% even at $1000$ data points)}, due to it being a higher dimensional problem (data requirements for estimating forward mapping grow exponentially in the dimension of $\boldtheta$~\cite{chen2016nonparametric}). Therefore, the error in the forward mapping compounded with the errors introduced due to the challenge of numerically inverting a deep network cause a significant drop in the performance of NI for our neuron-model examples.

\textbf{Limitations of our study}: Validation against real-world data is not performed. This is due to the lack of publicly available datasets for this problem. We aim to collect such datasets for future works (see Supp. Sec. 8). For simulations, we used single cell models of neurons. While existing studies~\cite{gopakumar2019cell} do use single-cell models to test their techniques, more complicated simulation studies using biological networks can be performed, although the rewards from doing more complicated simulation studies might be diminishing.  Active sampling techniques, which can offer substantial improvements, are not considered here and are a logical next step. Research on Density estimation is extensive, and many alternative techniques to MAF and MDN exist~\cite{papamakarios2019normalizing}, and should be explored. Our results are indicative of trends regarding the amount of data required between different techniques, and the absolute number of samples will depend upon many factors such as noise, the dimensionality/size of the waveform family being considered, etc. While our simulation results show that PATHFINDER outperforms existing methods when the dataset size is small in the scenarios investigated, the actual choice of technique will be influenced by the application of interest and needs to be treated on a case by case basis. 
\bibliography{Bibliography}




\end{document}